\documentclass[english,a4paper,nolineno]{lipics-v2019}
\newcommand{\R}{\mathbb{R}}

\newcommand{\nula}{\mathbf{0}}
\newcommand{\conv}{\mathop{\mathrm{conv}}}

\usepackage{microtype}

\title{The Crossing Tverberg Theorem}

\titlerunning{The Crossing Tverberg Theorem}

\author{Radoslav Fulek}{IST Austria, Klosterneuburg,  Austria}{radoslav.fulek@gmail.com}{https://orcid.org/0000-0001-8485-1774}{The author gratefully acknowledges support from Austrian Science Fund (FWF), Project~ M2281-N35.}

\author{Bernd G{\"a}rtner}{Department of Computer Science, ETH Z{\"u}rich, Switzerland}{gaertner@inf.ethz.ch}{https://people.inf.ethz.ch/gaertner/}{}

\author{Andrey Kupavskii}{University of Oxford, UK; Moscow Institute of Physics and Technology, Russia}{kupavskii@ya.ru}{http://kupavskii.com}{The research was supported by the Advanced Postdoc.Mobility grant no. P300P2\_177839 of the Swiss National Science Foundation}

\author{Pavel Valtr}{Department of Applied Mathematics, Faculty of Mathematics and Physics, Charles University, Prague, Czech Republic\\ and \\ Department of Computer Science, ETH Z{\"u}rich, Switzerland}{}{http://orcid.org/0000-0002-3102-4166, https://kam.mff.cuni.cz/\textasciitilde{}valtr/}{Research by P. Valtr was supported by the grant no. 18-19158S of the Czech Science Foundation (GA\v CR)}

\author{Uli Wagner}{IST Austria, Klosterneuburg, Austria}{uli@ist.ac.at}{https://orcid.org/0000-0002-3435-0100}{}

\authorrunning{R. Fulek and B. G{\"a}rtner and A. Kupavskii and P. Valtr and U. Wagner}

\Copyright{Radoslav Fulek and Bernd G{\"a}rtner and Andrey Kupavskii and Pavel Valtr and Uli Wagner}



\category{}

\relatedversion{}

\supplement{}

\funding{}

\acknowledgements{Part of the research leading to this paper was done during the 16th Gremo Workshop on Open Problems (GWOP), Waltensburg, Switzerland, June 12-16, 2018. We thank Patrick Schnider for suggesting the problem, and Stefan Felsner, Malte Milatz and Emo Welzl for fruitful discussions during the workshop. We also thank Stefan Felsner and Manfred Scheucher for finding and communicating the example from Section~\ref{sec:stronger}. We thank D{\"o}m{\"o}t{\"o}r P{\'a}lv{\"o}lgyi and the SoCG reviewers for various helpful comments.
}
\EventEditors{Gill Barequet and Yusu Wang}
\EventNoEds{2}
\EventLongTitle{35th International Symposium on Computational Geometry (SoCG 2019)}
\EventShortTitle{SoCG 2019}
\EventAcronym{SoCG}
\EventYear{2019}
\EventDate{June 18--21, 2019}
\EventLocation{Portland, United~States}
\EventLogo{socg-logo}
\SeriesVolume{129}
\ArticleNo{38} 

\nolinenumbers 
\hideLIPIcs  
\begin{document}

\maketitle

\begin{abstract}
The Tverberg theorem is one of the cornerstones of discrete geometry. It states that, given a set $X$ of at least $(d+1)(r-1)+1$ points in $\mathbb R^d$, one can find a partition $X=X_1\cup \ldots \cup X_r$ of $X$, such that the convex hulls of the $X_i$, $i=1,\ldots,r$, all share a common point. In this paper, we prove a strengthening of this theorem that guarantees a partition which, in addition to the above, has the property that the boundaries of full-dimensional convex hulls have pairwise nonempty intersections. Possible generalizations and algorithmic aspects are also discussed. 

As a concrete application, we show that any $n$ points in the plane in general position span $\lfloor n/3\rfloor$ vertex-disjoint triangles that are pairwise crossing, meaning that their boundaries have pairwise nonempty intersections; this number is clearly best possible. A previous result of Alvarez-Rebollar et al.\ guarantees $\lfloor n/6\rfloor$ pairwise crossing triangles. Our result generalizes to a result about simplices in $\mathbb R^d,d\ge2$.
 \end{abstract}

\section{Introduction and Results}
The following theorem was published by Johann Radon in 1921 \cite{Radon1921}: any set of $d+2$ points in $\mathbb R^d$ can be partitioned into two (disjoint) subsets, whose convex hulls intersect. In 1966, Helge Tverberg \cite{T1966_tverberg} proved the following important generalization of Radon's result.
\begin{theorem}[Tverberg~\cite{T1966_tverberg}]
\label{thm:Tverberg}
Let $X$ be a set of at least $(d+1)(r-1)+1$ points in $d$-space. Then $X$ can be partitioned into $r$ sets whose convex hulls all have a point $o$ in common.
(In the literature, the point $o$ is referred to as a \emph{Tverberg point} and the partition as a \emph{Tverberg partition}.)
\end{theorem}
Radon's theorem covers the case $r=2$.
Figure~\ref{fig:tverberg_partition} illustrates the Tverberg theorem for $d=2$ and $r=3$.
\begin{figure}[htb]
\includegraphics[width=\textwidth]{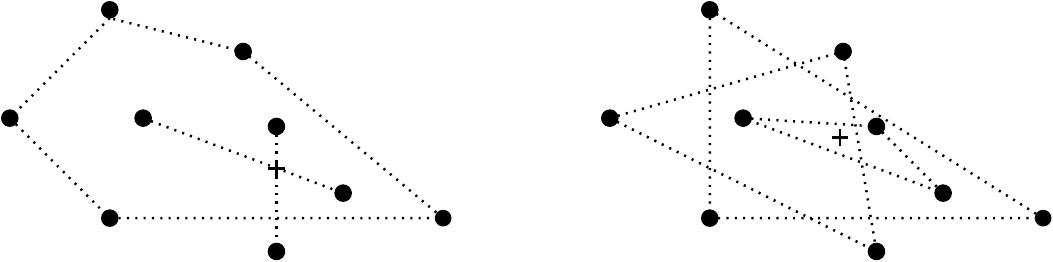}
\caption{The Tverberg theorem for $d=2$ and $r=3$: any set of at least $7$ points can be partitioned into three sets whose convex hulls all have a point in common. Our example uses $9$ points and shows two Tverberg partitions as well as corresponding Tverberg points. Many other Tverberg partitions exist in this example.\label{fig:tverberg_partition}}
\end{figure}

This theorem largely influenced the course of discrete geometry and spurred a lot of research in the area. We do not go into more details in this paper and refer the reader to a recent survey by B{\'{a}}r{\'{a}}ny and Sober{\'{o}}n~\cite{Barany2018}.

Throughout this paper, a set of points in $\R^d$ is said to be in \emph{general position} if no $d+1$ points are on a comon hyperplane (for example, no three points on a line in the plane).  

There is a major open question from the field of geometric graphs, motivating our work. In \cite{Aronov1994}, Aronov et al.\  conjectured that there exists an absolute constant $c>0$, such that, given any set of $n$ points in general position in the plane, one can find at least $cn$ disjoint pairs among them such that their connecting segments cross pairwise. Such a collection of segments is called a \emph{crossing family}. Despite considerable interest in this problem, the best published bound still comes from the original paper \cite{Aronov1994}, stating that one can always find a crossing family of size at least $c\sqrt n$ for some absolute $c>0$.
Only recently, after about 25 years, J.~Pach, N.~Rubin, and G.~Tardos improved this lower bound to a nearly linear one~\cite{PRT19_crossing}.

In another attempt to approach this problem, Alvarez-Rebollar et al.\ asked whether one can find at least $cn$ disjoint \emph{triples} whose connecting \emph{triangles} cross pairwise~\cite{RLU2016_crossingFam}. Throughout this paper, we say that two convex polytopes in $\R^d$ \emph{cross} if their boundaries have a non-empty intersection. We remark that if a convex polytope in $\R^d$ is not full-dimensional then it has no interior and therefore coincides with its boundary. Our main results below would be false, in the absence of general position, if relative boundaries were considered in the above definition of crossing (an easy counterexample in this situation is a set of points lying on a line in $\R^2$). 

We also point out that our definition of ``crossing'' is different from a classical one due to Fejes-T\'oth for general convex sets~\cite{FT67}. According to the latter, two convex sets $C,C'$ cross if neither $C\setminus C'$ nor $C'\setminus C$ are connected. Under general position, our notion of crossing convex polytopes generalizes the notion of a crossing of two vertex-disjoint segments in a crossing family (the two segments intersecting in their relative interiors)~\cite{Aronov1994}. It also generalizes the notion of Alvarez-Rebollar et al.\ for vertex-disjoint triangles (an edge of one triangle crossing an edge of the other)~\cite{RLU2016_crossingFam}.

Alvarez-Rebollar et al.\ showed the following: For every finite point set of size $n$ in the plane in general position, there exist $\left\lfloor\frac{n}{6}\right\rfloor$ vertex-disjoint and pairwise crossing triangles with vertices in $P$. As at most $\left\lfloor\frac{n}{3}\right\rfloor$ disjoint triangles can be found, this leaves a factor-2 gap.

If we only want triangles that have a common point, this gap can be closed using a simple strengthening of the Tverberg theorem for point sets of size at most $(d+1)r$ that we present next. However, these triangles might not be pairwise crossing, since triangles can be nested; see Figure~\ref{fig:tverberg_partition} (right).

\begin{theorem}\label{thm:Tverberg2}
Let $X$ be a set of at least $(d+1)(r-1)+1$ and at most $(d+1)r$ points in $d$-space. Then $X$ can be partitioned into $r$ disjoint sets $X_1,\ldots, X_r$ of size at most $d+1$, whose convex hulls all have a point in common. 
\end{theorem}

To prove this, we apply Theorem~\ref{thm:Tverberg} to $X$. We get sets $X'_1,\ldots, X'_r$ whose convex hulls contain a common point, say, the origin. Using Carath\' eodory's theorem, from every $X'_i$ of size larger than $d+1$ we can select $d+1$ points $X_i\subseteq X'_i$, whose convex hull still contains the origin. Finally, some of the sets $X'_i$ of size smaller than $d+1$ are filled up to size $d+1$ with the points removed from other $X'_i$'s.  The origin is still a common point for all $\conv(X_i)$'s.

The main result of this paper provides a crossing version of the Tverberg Theorem~\ref{thm:Tverberg2}.

\begin{theorem} \label{thm:main} Let $X$ be a set of at least $(d+1)(r-1)+1$ and at most $(d+1)r$ points in $d$-space. Then $X$ can be partitioned into $r$ disjoint sets $X_1,\ldots, X_r$ of size at most $d+1$, whose convex hulls all have a point in common. Moreover, for any $X_i,X_j$ of size $d+1$, $\conv (X_i)$ and $\conv (X_j)$ cross, meaning that their boundaries have a non-empty intersection.
\end{theorem} 

We call such a partition a \emph{crossing Tverberg partition}. An easy calculation shows that the number of sets with exactly $d+1$ elements is at least $|X|-dr\in\{r-d,r-d+1,\ldots,r\}$. In particular, for sets $X$ of size exactly $(d+1)r$, we immediately deduce the following simpler-looking corollary.

\begin{corollary}\label{cor1} Let $X$ be a set of $(d+1)r$ points in $d$-space. Then $X$ can be partitioned into $r$ sets $X_1,\ldots, X_r$ of size $d+1$, whose convex hulls all have a point in common and such that for any $i,j\in [n]$, $\conv (X_i)$ and $\conv (X_j)$ cross, meaning that their boundaries have a non-empty intersection.  
\end{corollary}

We also obtain an optimal strenghtening of the result by Alvarez-Rebollar et al.~\cite{RLU2016_crossingFam} that moreover generalizes to all dimensions $d\geq 2$.

\begin{corollary}
For every finite point set $X$ of size $n$ in the plane in general position, there exist $\lfloor n/3\rfloor$ vertex-disjoint and pairwise crossing triangles with vertices in $X$.

More generally, for every finite point set $X$ of size $n$ in $\mathbb R^d$ in general position, there exist $\lfloor n/(d+1)\rfloor$ vertex-disjoint and pairwise crossing simplices with vertices in $X$.
\end{corollary}
To derive this from Theorem~\ref{thm:main}, we remove $n\mod (d+1)$ points from $X$ and then apply Theorem~\ref{thm:main} on the remaining set of $(d+1)\left\lfloor\frac{n}{d+1}\right\rfloor$ points.

Finally, we get a crossing version of the actual Tverberg Theorem~\ref{thm:Tverberg}, for point sets of arbitrarily large size.

\begin{theorem}[Crossing Tverberg theorem]\label{thm:crossing_tverberg}
Let $X$ be a set of at least $(d+1)(r-1)+1$ points in $d$-space. Then $X$ can be partitioned into $r$ sets whose convex hulls all have a point in common. Moreover, for any $X_i,X_j$ of size at least $d+1$, $\conv (X_i)$ and $\conv (X_j)$ cross, meaning that their boundaries have a non-empty intersection.  
\end{theorem} 

This is also easy to prove, using Theorem~\ref{thm:main} and Corollary~\ref{cor1}. If $n:=|X|\leq (d+1)r$, we apply Theorem~\ref{thm:main}. Otherwise, we apply Corollary~\ref{cor1} to an arbitrary subset $Y\subset X$ of size $(d+1)r$, resulting in a crossing Tverberg partition into $r$ sets $Y_1,\ldots,Y_r$ of size $d+1$ each. Now we consecutively add the remaining points to suitable sets in such a way that the crossings between the convex hulls are maintained (the Tverberg point automatically remains valid). Suppose that some points have already been added, resulting in sets $X'_1,\ldots,X'_r$ whose convex hulls still cross pairwise. If the next point $p$ is contained in one of these convex hulls, we simply add it to the corresponding set. Otherwise, we select an inclusion-minimal set $\conv(X'_k\cup\{p\})$ among the (different) sets $\conv(X'_i\cup\{p\}), i=1,\ldots,r$, and add $p$ to $X'_k$.  We claim that this cannot result in nested convex hulls. For this, we only have to rule out that $\conv(X'_k\cup\{p\})$ ``swallows'' some other $\conv(X'_i)$. But this cannot happen, as otherwise $\conv(X'_i\cup\{p\})\subset \conv(X'_k\cup\{p\})$, contradicting the choice of $k$.

It remains to prove Theorem~\ref{thm:main}. 
We start with a Tverberg partition according to Theorem~\ref{thm:Tverberg2}, i.e.\ a partition into sets of size at most $d+1$ such that their convex hulls have a point in common. Such a partition might look like in Figure~\ref{fig:tverberg_partition2} (left). If the full-dimensional convex hulls (which are simplices) cross pairwise, we are done. In general, however, we still have pairs of nested simplices. As long as this is the case, we \emph{fix} one pair of nested simplices at a time until no pair of nested simplices exists anymore, and our desired crossing Tverberg partition is obtained.

By fixing, we mean that we repartition the $2(d+1)$ points involved in the two simplices in such a way that the resulting simplices are not nested anymore but still contain the common point. In the example of Figure~\ref{fig:tverberg_partition2} (left), there is one pair of nested simplices (red edges), and after fixing it (blue edges), we are actually done in this case; see Figure~\ref{fig:tverberg_partition2} (right).

\begin{figure}[htb]
\includegraphics[width=\textwidth]{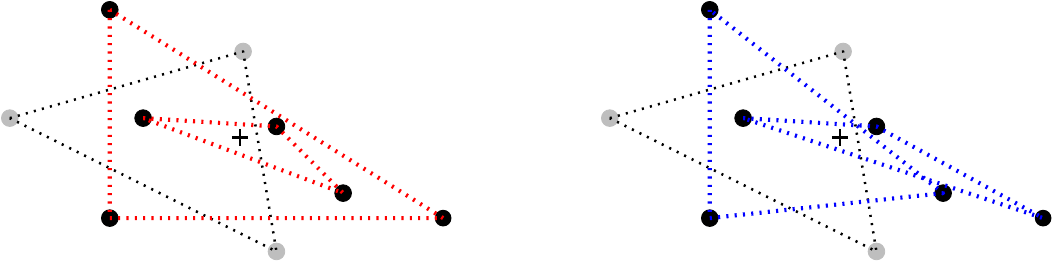}
\caption{Fixing a pair of nested simplices spanned by the 6 black points: two nested simplices (red) are transformed into two simplices that cross (blue).\label{fig:tverberg_partition2}}
\end{figure}


Repartitioning the points is always possible due to the following lemma, which may be of independent interest. This lemma is our main technical contribution, and we prove it in Section~\ref{sec:parity}.

\begin{lemma}\label{lem:unnesting}
Let $T,T'$ be two disjoint $(d+1)$-element sets in $\mathbb R^d$ such that $\nula \in \conv(T)\cap\conv(T')$. Then there exist two disjoint $(d+1)$-element sets $S,S'$ such that $S\cup S'=T\cup T'$, $\nula \in \conv(S)\cap\conv(S')$, and moreover, $\conv(S)$ and $\conv(S')$ cross.
\end{lemma}

To conclude the proof of Theorem~\ref{thm:main}, we still need to show that after finitely many fixing steps, there are no nested pairs anymore, in which case we have a crossing Tverberg partition.

To see this, we observe that in any fixing operation that replaces $X_i$ and $X_j$ such that $\conv(X_i)\subset \conv(X_j)$, the simplex $\conv(X_j)$ is volume-wise the unique largest $d$-dimensional simplex that can be formed from the $2(d+1)$ points $X_i\cup X_j$ involved in the operation. Hence, the two simplices $\conv(S)$ and $\conv(S')$ replacing $\conv(X_i)$ and $\conv(X_j)$ are volume-wise both strictly smaller than $\conv(X_j)$. Therefore, if we order all full-dimensional simplices by decreasing volume, the sequence of these volumes goes down lexicographically in every fixing operation.

Formally, let ${\cal V}=(V_1,V_2\ldots V_s), s\leq r$ be the sequence of volumes in decreasing order before the fix, and ${\cal V}'=(V'_1,V'_2\ldots V'_s)$ the decreasing order after the fix. Moreover, suppose that $k$ is the largest index at which the volume of $\conv(X_j)$ appears in ${\cal V}$. 
The volume of $\conv(X_i)$ appears at a position $\ell>k$. As the fixing operation removes a volume equal to $V_k$ and inserts two volumes smaller than $V_k$, we have that ${\cal V}$ and ${\cal V}'$ agree in the first $k-1$ positions, but $V'_{k}<V_k$. This exactly defines the relation ``${\cal V'}<{\cal V}$'' in decreasing lexicographical order, and as this is a total order, fixing must eventually terminate.

\vskip+0.1cm
\textbf{Remark. } Instead of volume, we may use the number of points from $X$ inside $\conv(X_j)$ as a measure. 

Since applications of Lemma~\ref{lem:unnesting} allow to keep the Tverberg point and the sizes of the parts in the (Tverberg) partition, we actually have the following strengthening of Theorem~\ref{thm:main}.

\begin{theorem}\label{thm:main-stronger} 
Let $X$ be a set of points in $d$-space.  Suppose that there is a Tverberg partition of $X$ into $r$ parts of sizes $s_1\dots,s_r$, for which $o$ is a Tverberg point, and $s_i\leq d+1, i=1,\ldots r$. Then there is also a \emph{crossing} Tverberg partition of $X$ into $r$ parts of sizes $s_1\dots,s_r$, for which a Tverberg point is $o$.  
\end{theorem}

In the next section, we prove Lemma~\ref{lem:unnesting}. We remark that for each of the other theorems and corollaries in this section, we have already explained how to derive them from the Tverberg theorem, or from our main Theorem~\ref{thm:main}. In Section~\ref{sec3}, we discuss possible generalizations as well as some algorithmic aspects of the problem.

\section{Proof of Lemma~\ref{lem:unnesting}}\label{sec:parity}
We employ the following parity lemma that we prove in the next two subsections (in a geometric way for $d=2$, and in a combinatorial way for all dimensions).

\begin{lemma}[Parity lemma]\label{lemparity} Let $V\subset \mathbb R^d$, $|V|=2(d+1)$, such that $V\cup\{\nula\}$ is in general position. Then
$$\Big|\Big\{\{F,G\}:F,G\in {V\choose d+1}, F\cap G=\emptyset, \nula\in \conv F\cap \conv G\Big\}\Big| \ \ \ \text{is even.}$$
\end{lemma}
In particular, if there is one such pair $\{F,G\}$, then there is another.

Before we move on to the proof, let us derive Lemma~\ref{lem:unnesting} from this. It suffices to handle the case where the set $T\cup T'\cup\{\nula\}$ is in general position. To justify this, we make two observations. The first is that general position is achieved by \emph{some} arbitrarily small (for example, random) perturbation of $T\cup T'$. The second is that the property of \emph{not} having sets $S,S'$ as required is maintained under \emph{any} sufficiently small perturbation. Putting the two observations together, we see that if Lemma~\ref{lem:unnesting} holds for $T\cup T'\cup\{\nula\}$ in general position, then it holds for all $T,T'$.

Then, by Lemma~\ref{lemparity} with $V=T\cup T'$, there is another pair $\{S,S'\}$ such that $S\cup S'=T\cup T'$, $\nula \in \conv(S)\cap\conv(S')$. At most one of these pairs can yield nested simplices: the outer simplex of a nested pair coincides with $\conv(T\cup T')$ and is therefore uniquely determined. Hence, one of the two pairs yields crossing simplices. 

\vskip+0.1cm
\textbf{Remark. } We note that this statement and its application is quite unusual. Typically, one proves that the number of objects of a certain type is always odd and thus at least one object of the type exists.

\subsection{Geometric proof of the parity Lemma~\ref{lemparity} for $\mathbf{d=2}$}
The purely combinatorial proof that we give below for all $d$ is not difficult, but does not provide any geometric intuition. We therefore start with a simple proof in the plane.

Consider a set $V$ of $2(d+1)=6$ points in the plane. We remark that the statement is invariant under scaling points, and thus we may assume that all points lie on a circle with the center in the origin. Due to general position, no two points from $V$ lie on a line passing through the origin; see Figure~\ref{fig:geomproof} (left).

\begin{figure}[htb]
\begin{center}
\includegraphics[width=0.6\textwidth]{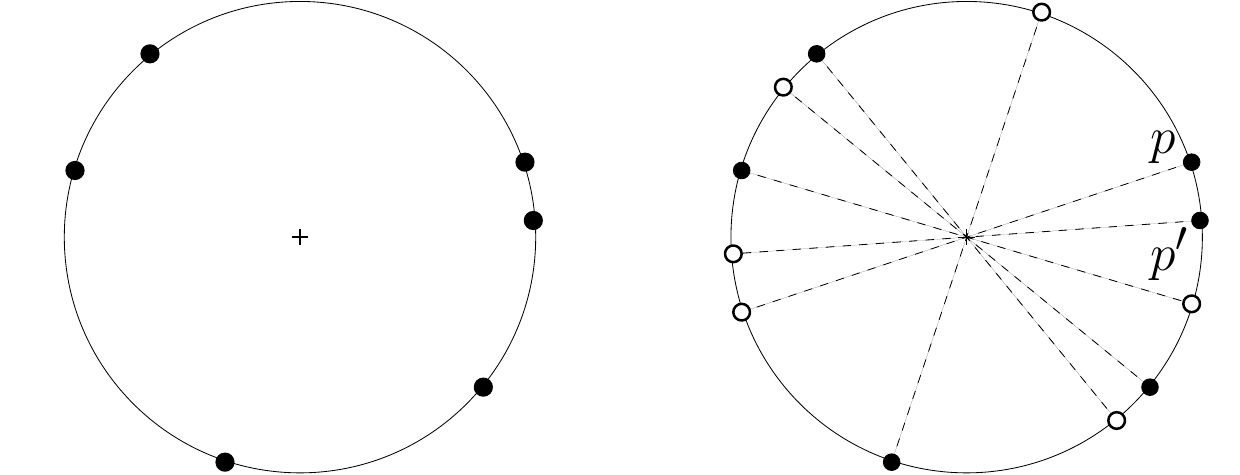}
\caption{6 points in the plane (black) and their mirror images (white); there must be two consecutive black points $p,p'$.\label{fig:geomproof}}
\end{center}
\end{figure}

Now we reflect each point at the origin and obtain another $6$ points, drawn in white in  Figure~\ref{fig:geomproof} (right). We observe that in the circular order of points, there must be two consecutive ones of the same color. Indeed, an alternating pattern (black, white, black, white,\ldots) would lead to pairs of antipodal points having the \emph{same} color. Let $p,p'$ be two consecutive points of the same color; by going to the antipodal points if necessary, we may assume that they are black and hence belong to $V$. 

We make two observations (actually, just one). (i) $p$ and $p'$ cannot belong to a triple $F\subset V$ such that $\nula\in\conv(F)$, as otherwise, the third point would get reflected to a white point between $p$ and $p'$; see Figure~\ref{fig:geomproof2} (left). (ii) For any $q\in V\setminus\{p,p'\}$, the two segments $\conv(\{p,q\})$ and $\conv(\{p',q\})$ have the origin on the same side; see Figure~\ref{fig:geomproof2} (right). 

\begin{figure}[htb]
\begin{center}
\includegraphics[width=0.6\textwidth]{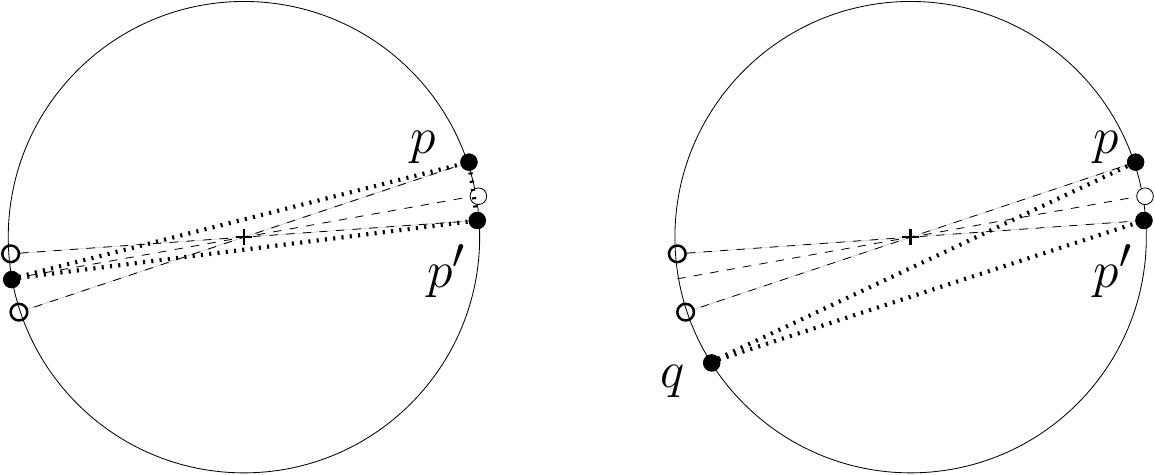}
\caption{Consecutive black points are combinatorially indistinguishable.\label{fig:geomproof2}}
\end{center}
\end{figure}

This implies the following: if $\{F,G\}$ is a partition of $V$ that we count in Lemma~\ref{lemparity}, then (i) $p$ and $p'$ are in different parts, and (ii) swapping $p$ and $p'$ between the parts leads to a different partition $\{F',G'\}$ that we also count. In other words, $p$ and $p'$ are ``combinatorially indistinguishable'' with respect to the relevant properties, and the operation of swapping them between parts establishes a matching between the partitions that we want to count. Hence, their number is even.

\subsection{Combinatorial proof of the parity Lemma~\ref{lemparity}}
Recall that for $|V|=2(d+1)$, we need to show that there is an even number of partitions $\{F,G\}$ of $V$ into parts of equal size $d+1$ such that $\nula\in\conv(F)\cap\conv(G)$.
We show that this follows from the fact that the $(d+1)$-element subsets $F$ with $\nula\in\conv(F)$ form a \emph{cocycle}, a concept borrowed from topology. 

\begin{definition} Let $n\ge k\ge 1$ be integers, and let $V$ be a set with $n$ elements. A family $\mathcal C\subset {V\choose k}$ of $k$-element subsets of $V$ is a {\rm cocycle} if $$|\{F\in \mathcal C: F\subset M\}| \ \ \text{is even for every } M\subset V\ \ \text{with } \ |M|=k+1. $$\end{definition}

\textbf{Example. } Fix a set $D\in {V\choose k-1}$. Then $\delta D:=\{F\in {V\choose k}: D\subset F\}$ is a cocycle. Indeed, a $(k+1)$-element subset $M\subset V$ either contains no sets in $\delta D$ (if $D\not \subset M$) or exactly two sets in $\delta D$ (If $D=M\setminus\{p,q\}$).

\begin{lemma}\label{lem1}
  Let $V\subset \mathbb R^d$ be such that $V\cup\{\nula\}$ is in general position. Then
  $$\mathcal C(V):=\big\{F\in {V\choose d+1}: \nula \in \conv F\big\}$$
  is a cocycle.
\end{lemma}
\begin{proof}
Let $M:=\{v_1,\ldots,v_{d+2}\}\subset V$. Lift the points to dimension $d+1$ such that the convex hull of the lifted point set $\hat{M}$ is a full-dimensional simplex $\Delta$; see Figure~\ref{fig:lifting}.

\begin{figure}[htb]
\begin{center}
\includegraphics[width=0.6\textwidth]{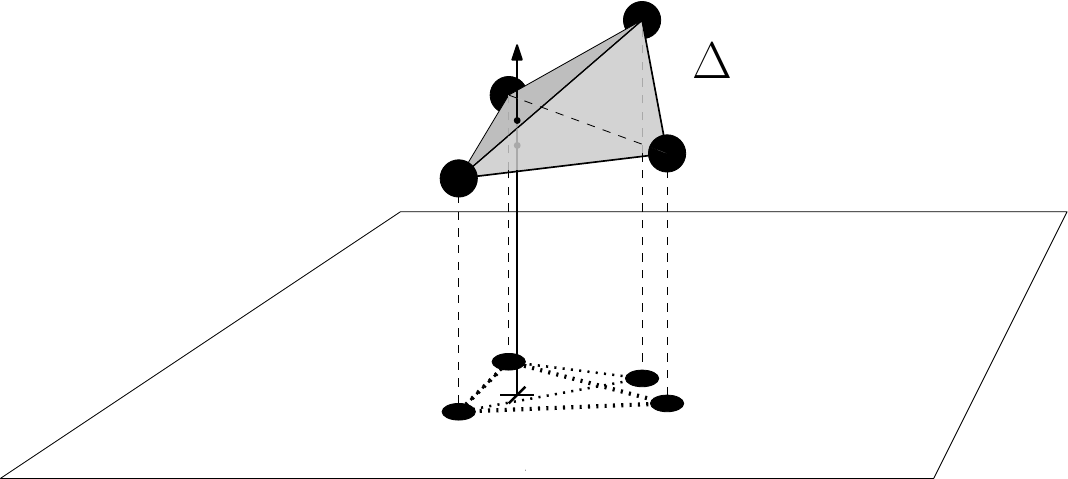}
\caption{Illustration of the geometric proof of Lemma~\ref{lem1} when $d=2$\label{fig:lifting}}
\end{center}
\end{figure}

Any set $\hat{F}\subset \hat{M}$ of size $d+1$ spans a facet of $\Delta$, and we have $\nula\in\conv(F)$ if and only if the vertical line through $\nula$ intersects that facet. As $V\cup\{\nula\}$ is in general position, this line intersects the convex set $\Delta$ in a line segment, if it intersects it at all; hence, the line intersects the boundary of $\Delta$ in zero or two points, each lying in the interior of some facet. Thus, $|\{F\in \mathcal C(V): F\subset M\}|\in\{0,2\}$.

We remark that there is also an elementary linear algebra version of this ``proof by picture'' (omitted in this extended abstract). 
\end{proof}

By Lemma~\ref{lem1}, Lemma~\ref{lemparity} is now simply a special case of the following main result of this section.

\begin{theorem}\label{prop2} Let $k\ge 1$, $|V|=2k$ and let $\mathcal C\subset {V\choose k}$ be a cocycle. Set    $$\mathcal P_{\mathcal C}:=\big\{\{F,G\}: F,G\in \mathcal C, F\cap G=\emptyset \big\}.$$
Then $|\mathcal P_{\mathcal C}|$ is even.
\end{theorem}

For the proof, we need one more concept. For a family $\mathcal D\subset {V\choose k-1}$, we define the \emph{coboundary} 
\[\delta\mathcal D=\big\{F\in {V\choose k}: F\ \text{contains an odd number of sets }D\in \mathcal D\big\}.\] This naturally extends the definition $\delta D=\{F\in {V\choose k}: D\subset F\}$ for $D\in {V\choose k-1}$ from our previous example. We call $\delta D$ the \emph{elementary cocycle} associated with $D$. 

To prove Theorem~\ref{prop2}, we use the following basic fact.

\begin{lemma}\label{lem2}
For every cocycle $\mathcal C\subset {V\choose k}$ there exists a family $\mathcal D\subset {V\choose k-1}$ such that
$\mathcal C=\delta\mathcal D$.
Equivalently,
$$\mathcal C=\delta D_1\oplus\delta D_2\oplus \ldots \oplus \delta D_m,$$
where $\mathcal D=\{D_1,\ldots, D_m\}$ and $\oplus$ denotes symmetric difference.  
\end{lemma}

\begin{proof}[Proof of Theorem~\ref{prop2}]
The statement is trivially true when $\mathcal C = \emptyset$ is the empty cocycle.

Thus, by Lemma~\ref{lem2}, it suffices to show that the parity of $|\mathcal P_{\mathcal C}|$ does not change when we take the symmetric difference with an elementary cocycle, i.e., 
\begin{equation}\label{eqparity}|\mathcal P_{\mathcal C}|\equiv |\mathcal P_{\mathcal C\oplus\delta D}|\mod 2\end{equation}
for any $D\in {V\choose k-1}$.
\begin{claim} $\mathcal P_{\mathcal C\oplus\delta D}=\mathcal P_{\mathcal C}\oplus \mathcal R,$ where
$$\mathcal R:=\big\{\{F,G\}: F\in \delta D, G\in \mathcal C, F\cap G=\emptyset \big\}.$$
\end{claim}
\begin{proof}
To see this, note that for pairs $\{F,G\}$ with $D\not \subset F$ and $D\not\subset G$ nothing changes, i.e., such a pair is either in both $\mathcal P_{\mathcal C}$ and $\mathcal P_{\mathcal C\oplus\delta D}$ or in neither. Thus, it suffices to consider pairs $\{F,G\}$, $F\cap G=\emptyset$, such that $D$ is contained in one of the two sets. Without loss of generality, $D\subset F$ and $G\subset V\setminus D$.

Since $D\not\subset G$, we have $G\in \mathcal C$ if and only if $G\in \mathcal C\oplus\delta D$.  Since $D\subset F$, we have $F\in \delta D$, and hence, $F\in \mathcal C$ if and only if $F\not\in \mathcal C\oplus \delta D$. Thus, $\{F,G\} \in \mathcal P_{\mathcal C\oplus\delta D}$ if and only if  $\{F,G\}\in \mathcal{R}$ and  $\{F,G\}\not\in \mathcal{P}_{\mathcal{C}}$. This proves the claim.
\end{proof}

Since $\mathcal C$ is a cocycle and $|V\setminus D|=k+1$, there is an even number of $G\in \mathcal C$ with $G\subset V\setminus D$. Since $|V|=2k$, each such $G$ determines a unique $F=V\setminus G$, $F\in \delta D$, with $\{F,G\}\in \mathcal R$. Thus,
$$|\mathcal R|=\big|\{G\in \mathcal C: G\subset V\setminus D\}\big|\equiv 0\mod 2.$$
Combined with the claim, we get that \eqref{eqparity} holds, which concludes the proof of the theorem.
\end{proof}

For the sake completeness, we also prove Lemma~\ref{lem2}.

\begin{proof}[Proof of Lemma~\ref{lem2}] Let $\mathcal C\subset {V\choose k}$ be a cocycle. Choose an arbitrary element $v\in V$. Define 
$$\mathcal D:=\mathcal D_v:=\big\{D\in{V\setminus v\choose k-1}: D\cup \{v\}\in \mathcal C\big\}.$$
 
We claim that $$\mathcal C=\delta \mathcal D.$$
Let $F\in {V\choose k}$. We distinguish two cases:
\begin{itemize}
  \item[1) ] If $v\in F$ then, by the definition of $\mathcal D\subset {V\setminus \{v\}\choose k-1}$, we have $F\setminus \{v\}\in \mathcal D$ if and only if $F\in \mathcal C$. Moreover, $F\setminus \{v\}$ is the only set in ${V\setminus \{v\}\choose k-1}$ that $F$ contains. Thus,
$$F\in \mathcal C\ \Leftrightarrow\ F\in \delta \mathcal D\ \ \ \text{in this case.}$$
  \item[2) ] Assume $v\notin F$. Then $M:=F\cup \{v\}$ has $k+1$ elements. Since $\mathcal C$ is a cocycle, $M$ contains an even number of sets from $\mathcal C$. Thus,
\begin{align*}
  F\in \mathcal C & \Leftrightarrow M \ \text{ contains an odd number of }G\in \mathcal C \ \text{with } v\in G \\
   & \Leftrightarrow F\ \text{ contains an odd number of sets }D=G\setminus \{v\}\in \mathcal D\\
   & \Leftrightarrow F\in \delta \mathcal D.
\end{align*}
\end{itemize}
\end{proof}

\section{Discussion}\label{sec3}
\subsection{A Topological version of Theorem~\ref{thm:main}}

If $r$ is a prime power, then the Tverberg theorem admits a topological generalization, known as the \emph{Topological Tverberg theorem}:
\begin{theorem} 
\label{thm:topTverberg}
If $d\in \mathbb{N}$ and if $r$ is a prime power, then for every continuous map from the $(d+1)(r-1)$-dimensional simplex $\Delta_{(d+1)(r-1)}$ to $\R^d$, there exist $r$ pairwise disjoint faces of $\Delta_{(d+1)(r-1)}$ whose images intersect in a common point. 
\end{theorem}
This result was first proved in the case when $r$ is a prime by {B{\'a}r{\'a}ny, Shlosman and Sz{\"u}cs~\cite{BSS1981_topTverbergPrime} and later extended to prime powers by \"{O}zaydin~\cite{O1987_topTverberg}. On the other hand, Theorem~\ref{thm:topTverberg} is false if $r$ is not a prime power: By work of Mabillard and Wagner \cite{MabillardWagner2015} and a result of \"{O}zaydin~\cite{O1987_topTverberg}, a closely related result (the \emph{generalized Van Kampen--Flores theorem}) is false whenever $r$ is not a prime power and, as observed by Frick~\cite{F15_approx}, the failure of  Theorem~\ref{thm:topTverberg} for $r$ not a prime power follows from this by a reduction due to Gromov \cite{Gromov} and to Blagojevi\'{c}, Frick, and Ziegler \cite{BFZ_Constraints_2014}. The lowest dimension in which counterexamples are known to exist is $d=2r$~\cite{AMSW15_codim2,MabillardWagner2015}.
We refer to the recent surveys~\cite{Barany2018,BZ2017_story,S18_user,Z18} for more background on the Topological Tverberg theorem and its history.

In the same vein, it is natural to wonder if our Theorem~\ref{thm:main} also extends to the topological setting. The straightforward approach to generalize our result would be to use the Topological Tverberg theorem instead of the Tverberg theorem and then keep fixing the pairs of simplices whose boundaries do not mutually intersect. To this end we need an adaptation of  Lemma~\ref{lem1}  and the fixing procedure
to the topological setting. The rest of our argument is free of any geometry except for the use of Carath\'{e}odory's theorem which is not really crucial. While extending  Lemma~\ref{lem1} is easy, showing the termination of the fixing procedure appears to be quite difficult except under the scenario which we discuss below.
First, we discuss planar extensions of Theorem~\ref{thm:main}.
An \emph{arrangement of pseudolines} $\mathbb{P}$ is a finite set of not self-intersecting open arcs, called \emph{pseudolines}, in $\mathbb{R}^2$ such that (i) For every pair $P_1,P_2\in \mathbb{P}$ of two distinct pseudolines, $P_1$ and $P_2$ intersect transversely in a single point, and (ii) $\mathbb{R}^2\setminus P$ is not connected for every $P\in \mathbb{P}$.
A drawing of a complete graph on $n$ vertices $K_n$ in the plane is \emph{pseudolinear} if the edges can be extended to an {arrangement of pseudolines}.

In the plane, it is not hard to see that  Theorem~\ref{thm:main} and its proof almost extends to the setting of  pseudolinear drawings of complete graphs. Since the Topological Tverberg theorem is only valid for prime powers $r$, the number of pairwise crossing triangles is slightly smaller than the number of vertices divided by 3.

\begin{theorem} \label{thm:mainTopPlane}
In a pseudolinear drawing of a complete graph  $K_{3n}$  we can find $m=(1-o(1))n$ vertex-disjoint and pairwise crossing triangles. Moreover, the topological discs bounded by these triangles intersect in a common point.
\end{theorem}

\begin{proof}
Let $\mathcal{D}$ be a pseudolinear drawing of  $K_{3n}$.
We put $m$ to be the largest prime power not larger than  $n$.
By the asymptotic law of distribution of prime numbers  $m=(1-o(1))n$.
Next, apply the Topological Tverberg theorem with $r=m$ and $d=2$ to a map $\mu:\Delta_{3m-3}\rightarrow \mathbb{R}^2$
which extends $\mathcal{D}$  as follows.
We define $\mu$ on the 1-dimensional skeleton of $\Delta_{3m-3}$ as  a restriction of $\mathcal{D}$  to some $K_{3m-2}$.
Note that every triangle of $K_{3m-2}$ is drawn by $\mathcal{D}$ as a closed arc without self intersections. The map $\mu$ extends to the 2-dimensional skeleton of $\Delta_{3m-3}$ so that every  2-dimensional face is mapped homeomorphically in $\mathbb{R}^2$. We  define the map $\mu$ on the rest of $\Delta_{3m-3}$ arbitrarily while maintaining continuity.

Analogously to the proof of Theorem~\ref{thm:main}, an application of the Topological Tverberg theorem gives us $m-1$ disjoint $2$-dimensional faces $F_1,\ldots, F_{m-1}$ of $\Delta_{3m-3}$ whose images under $\mu$ intersect in a common point.
 To this end we apply Carath\'eodory's theorem for drawings of complete graphs~\cite[Lemma 4.7]{BFK15_monotone} instead of the original version of Carath\'eodory's theorem.
Let $T_i$ denote the boundary of $F_i$, for $i=1,\ldots, m-1$.
Note that each $T_i$ is a triangle in $K_{3m-2}$.
If all pairs $T_i$ and $T_j$ are crossing then we are done. Otherwise, we perform the fixing operations, which can be done since  Lemma~\ref{lem1} easily extends to the setting in which  we replace simplices  by images of 2-dimensional faces of $\Delta_{3m-3}$  under $\mu$.

The procedure of applying successively the fixing operation terminates due to the following argument.
First, observe that 
every four vertices in a pseudolinear drawing of a complete graph induce
either a crossing free drawing of $K_4$ or a drawing of $K_4$ with exactly one pair of crossing edges in the interior of a disc bounded by a crossing free 4-cycle, see Figure~\ref{fig:pseudo} for an illustration.
It follows that if $\mu(T_i) \cap \mu(T_j)=\emptyset$, and let's say $\mu(F_i)\subset \mu(F_j)$, then the restriction of $\mu$ to the subgraph of $K_{3m-2}$  induced by $V(T_i)\cup V(T_j)$ is contained in $\mu(F_j)$.
Indeed, otherwise there exists a vertex $u\in V(T_i)$ and $v\in V(T_j)$ such that $\mu(uv)$ crosses an edge of $T_j$, let us denote it by $wz$, that is not incident to $v$. Then the restriction of $\mu$ to the subgraph of $K_{3m-2}$ induced by $\{u,v,w,z\}$ is a drawing of $K_4$ that is not pseudolinear  (contradiction). Therefore the volume argument goes through if we consider, say, volumes of $\mu(F_i)$'s.
\end{proof}

\begin{figure}
\centering
\includegraphics[scale=1]{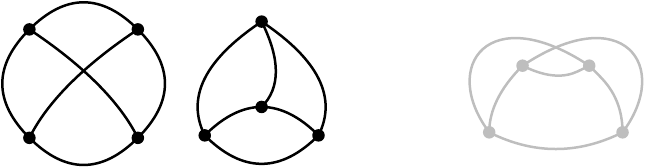}
\caption{The allowed drawings of $K_4$ in a pseudolinear drawing of $K_n$ (left). The forbidden drawing of $K_4$ in a pseudolinear drawing of $K_n$ (right).\label{fig:pseudo}}
\end{figure}

A drawing of a  graph in the plane is \emph{simple} if every pair of edges intersect at most once either at a common end point or in a proper crossing.
Clearly, all pseudolinear drawings of complete graphs are also simple, but not vice-versa, as illustrated in the last drawing in Figure~\ref{fig:pseudo}. Hence, it might be worthwhile to extend   Theorem~\ref{thm:mainTopPlane} to simple drawings of complete graphs.
 
If we want the interiors of the triangles to be pairwise intersecting, we only known that we can take $m$ at least $\Omega(\log^{1/6}n)$ which is easily derived from the following result of Pach, Solymosi and T\'oth~\cite{PST2003_unavoidable}. Every simple drawing of $K_n$ contains a drawing of $K_m$ that is weakly isomorphic to  a so-called convex  complete graph or a twisted complete graph, see Figure~\ref{fig:twisted}, for which Theorem~\ref{thm:mainTopPlane} holds. We omit the proof of the latter which is rather straightforward. For example, if in a twisted drawing of $K_{3n}$ the vertices are labeled as indicated in the figure, $\{\{0+i,n+i,2n+i\}| \ i=0,\ldots, n-1\}$ is a crossing Tverberg partition.

 If we do not insist on the interiors of the triangles to be pairwise intersecting, we know that $m$ can be taken to be at least $\Omega(n^{\varepsilon})$ for some small $\varepsilon>0$ by the following result of Fox and Pach~\cite{FP2012_coloring}.
Every simple drawing of $K_n$ contains $\Omega(n^{\varepsilon})$ pairwise crossing edges for some $\varepsilon>0$.

\begin{figure}

\centering
\includegraphics[scale=0.7]{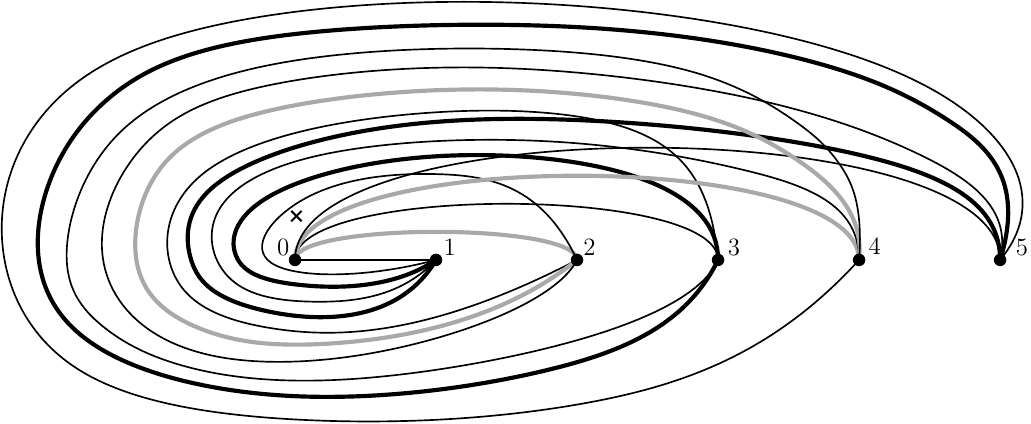}
\caption{The twisted drawing of $K_6$. The bold grey and black edges form boundaries of two crossing triangles. The symbol x marks a  point contained in the interior of these two triangles.\label{fig:twisted}}
\end{figure}

Theorem~\ref{thm:mainTopPlane} could be generalized to hold for an appropriate high dimensional analog of pseudolinear drawings. Since this would require  introducing  many technical terms and would not offer substantially interesting content we refrain from doing so.

\subsection{Stronger conditions on crossings}\label{sec:stronger}
%

A pair of vertex disjoint $(\lceil d/2 \rceil-1)$-dimensional simplices in general position in $d$-space does not intersect. 
Hence, the pairwise intersection of the boundaries of $\conv(X_i)$'s in the conclusion of   Theorem~\ref{thm:main}, cannot be strengthened to the pairwise intersection of lower than
$\lceil d/2 \rceil$-dimensional skeleta of the boundaries. 

Nevertheless, for $d=3$ one may ask if in the setting of Lemma~\ref{lem:unnesting}, we get a stronger property along the following lines. 
Can we guarantee the existence of a pair of vertex-disjoint tetrahedra $\{S,S'\}$ that both contain the origin and such that the boundary of a 2-dimensional face $F$  of $S$ is linked with the boundary of a 2-dimensional face $G$ of $S'$, meaning that the boundary of $F$ intersects $G$ and the boundary of $G$ intersects $F$? 
Again, the answer to this question is negative. Stefan Felsner and Manfred Scheucher (personal communication) found the following set of 8 points:
      $$(3, -2, 2),\ \ \ \
      (2, -5, 3),\ \ \ \
      (-3, 0, -4),\ \ \ \
      (-1, 2, 0),$$
      $$(1, -5, -4),\ \ \ \
      (4, 1, -2),\ \ \ \
      (-2, -5, -4),\ \ \ \
      (-3, 1, 3).$$
     This set determines a pair of disjoint tetrahedra both containing the origin $(0,0,0)$, but no two disjoint linked tetrahedra
     both containing the origin.

The example was found using a SAT solver who found an abstract order type with the required property. A realization of the order type with actual points was obtained with a randomized procedure.

\subsection{Computational complexity of finding a crossing Tverberg partition}
A natural question is whether we can find the partition of the point set given by Theorem~\ref{thm:main} efficiently, i.e.,  in polynomial time in the size of $X$. 
A straightforward way to construct an algorithm is to  make the proof of Theorem~\ref{thm:main} algorithmic.
To this end we first need an algorithm for finding a Tverberg partition and also a Tverberg point. Unfortunately, it is neither known whether this can be done in polynomial time, nor are there hardness results~\cite[Section 3.4]{TverbergComplexity18}. Since a Tverberg partition always exists, the decision problem is trivial, so NP-completeness cannot apply. It might still be possible to prove completeness of the problem for a suitable subclass of TFNP (containing search problems for which a solution is guaranteed to exist). It has been shown that the problem is contained in two such subclasses that have complete problems, namely PPAD and PLS~\cite[Theorem 4.9]{MMSS17_endofline}. While this makes it unlikely that the problem is complete for either of them, it could well be contained in and turn out to be complete for a subclass of PPAD$\cap$PLS, a number of which have been introduced and studied in recent years; see~\cite{DBLP:journals/corr/abs-1811-03841} and the references therein.


The only closely related hardness result we are aware of is the one by Teng~\cite[Theorem 8.14]{T92_points}, who proved  that checking whether a given point is a Tverberg point of a given point set is NP-complete.

A line of research on finding an approximate Tverberg partition efficiently was initiated by Miller and Sheehy~\cite{MS10_approx} and further developed in~\cite{MW13_approx,RS16_approx}.
In particular, Mulzer and Werner~\cite{MW13_approx} showed that it is possible to find  in time $d^{O(\log d)|X|}$ an approximate Tverberg partition of size $\left\lceil \frac{|X|}{4(d+1)^3}\right\rceil$, whereas Theorem~\ref{thm:Tverberg} guarantees the partition of size $\left\lceil \frac{|X|}{d+1}\right\rceil$.

If we aim only at an approximate algorithmic version of our Theorem~\ref{thm:main} along the lines of the result of Mulzer and Werner, we face the problem of efficiently fixing an (approximate) Tverberg partition to make it crossing. Due to the fact that our termination argument for the iterated \emph{Fixing Pairs} procedure relies on progress in the lexicographical ordering of the simplex volumes, we may potentially need exponentially many (in the size of~$X$) iterations before we arrive at a crossing partition. 
We leave it as an interesting open problem to prove or disprove that there is always a way to invoke the \emph{Fixing Pairs} operation only polynomially (or least subexponentially) many times in order to arrive at a partition required by Theorem~\ref{thm:main}.



\bibliography{paper_references}

\end{document}